\documentclass[a4paper,10pt]{article}
\usepackage[utf8]{inputenc}

\usepackage{amssymb}
\usepackage{amsmath} 
\usepackage{amsthm}
\usepackage{url}

\newtheorem{proposition}{Proposition}[section]
\theoremstyle{definition}
\newtheorem{definition}[proposition]{Definition}
\newtheorem{notation}[proposition]{Notation}

\title{Bisimulation and p-morphism for branching-time logics with indistinguishability relations}
\author{Alberto Gatto\\
\footnotesize Department of Computing, Imperial College London, London SW7 2AZ, UK\\
\footnotesize \url{http://www.doc.ic.ac.uk/~ag2512/}\\
\footnotesize February 2013}
\date{}

\begin{document}

\maketitle

\begin{abstract}
In Zanardo, 1998,
the Peircean semantics
for branching-time logics is enriched with a notion of
indistinguishability at a moment $t$ between histories
passing through $t$.
Trees with indistinguishability relations provide a semantics for
a temporal language with tense and modal operators.
In this paper a notion of p-morphism and a notion of bisimulation,
wrt this language and semantics,
are given and a number of preservation results are proven.\\

\end{abstract}

\section*{Introduction}
Various work on logics of agency enriches the Peircean semantics
for branching-time logics with a notion of
undividedness at a moment $t$ between histories
passing through $t$, e.g. \cite[Belnap et al., 2001]{bpx01}.
In \cite[Zanardo, 1998]{zan98}, undividedness is generalized by the notion of indistinguishability.
Trees with indistinguishability relations provide a semantics for
a temporal language with tense and modal operators.
In this paper, in \S\ref{Preliminaries}, a language without the ``weak future operator''
(for every history there is a point in the future) is considered and
an alternative view of the semantics is presented.
Then, in \S\ref{bsim and pmorph}, a notion of p-morphism and a notion of bisimulation, wrt this language and semantics,
are given and a number of preservation results are proven.
Finally, in \S\ref{adding F}, the language is enriched with the ``weak future operator'',
a notion of p-morphism and a notion of bisimulation, wrt this language and semantics, are given and
a number of preservation results are proven.\\

\section{Preliminaries}\label{Preliminaries}
In this section, the syntax and the semantics are introduced.
After that, a different view of the semantics is presented.
The idea underlying this different view of the semantics
is not new, see e.g. \cite[\S3]{zan13}.
Behind this different view of the semantics, there is a different view of trees.
The idea underlying this different view of trees is not new, see e.g. \cite[\S4]{tho84}.

\subsection{Syntax}
Here, the language and what is a formula are defined.
\begin{definition}\label{language}
Let $PV$ be a denumerable set.
The elements of $PV$ are called \textit{propositional variables} or \textit{atoms}.
The set $\mathfrak{L}=PV\cup\{(,),\neg,\wedge,G,H,L\}$ is called \textit{language}.
\textit{Formulas} are strings of elements of the language
built up recursively according to the following rules:
\begin{enumerate}
\item For every $p\in PV$, $p$ is a formula.
\item If $\varphi$ and $\psi$ are formulas,
$(\neg\varphi)$, $(\varphi\wedge\psi)$, $(G\varphi)$, $(H\varphi)$ and
$(L\varphi)$ are formulas.
\end{enumerate}
$\vee$ and $\rightarrow$ are
the usual abbreviations. $P$ abbreviates $\neg H\neg$, $f$
abbreviates $\neg G\neg$
and $M$ abbreviates $\neg L\neg$.
The usual precedence rules among operators are assumed.
\end{definition}

\subsection{Semantics}\label{semantics}
Here, a number of definitions
are given in order to define the semantics. After them, satisfiability and validity
for a formula wrt a frame and wrt a model are defined.

\begin{definition}\label{DownwardLinear}
A binary relation $R$ over a set $A$ is said \textit{downward linear} provided, for
each $a,b,c\in A$ such that $bRa$ and $cRa$, $b=c$ or $bRc$ or $cRb$.
\end{definition}

\begin{definition}\label{Tree}
A \textit{tree} is a 2-tuple $(  T,< )$, where $T$ is a
set and $<$ is an irreflexive, transitive and downward linear binary
relation on $T$.
\end{definition}

\begin{definition}\label{History}
Given a tree $\mathfrak{T}=(  T,< )$, an \textit{history} of $\mathfrak{T}$ is an
$\subseteq$-maximal $<$-linear subset of $T$.
$H_{\mathfrak{T}}$ denotes the set of histories of $\mathfrak{T}$.
Given $t\in T$, $H_{\mathfrak{T},t}$ denotes the set of histories $h$
in $\mathfrak{T}$ passing through $t$, i.e. with $t\in h$.
\end{definition}

\begin{definition}\label{IndistinguishabilityFunction}
Given a tree $\mathfrak{T}=(  T,< )$, a
function $I:T\rightarrow\mathfrak{P}(H_{\mathfrak{T}}\times H_{\mathfrak{T}})$,
$t\mapsto I_t$, is called
\textit{indistinguishability function} if, for every $t\in T$, $I_t$ fulfills the following
conditions: 
\begin{enumerate}
\item $I_t$ is an equivalence relation over $H_{\mathfrak{T},t}$.
\item\label{Indistinguishability Condition} For every $s\in T$ and $h,k\in H_{\mathfrak{T},t}$,
if $hI_tk$ and $s<t$ then $hI_sk$.
\end{enumerate}
Given $t\in T$, $\Pi_{\mathfrak{T},t}$ denotes the set of the equivalence classes
of $I_t$. 
\end{definition}
The suffixes will be forgotten when there is no case of confusion.

\begin{definition}
A 3-tuple $\mathfrak{F}=
(  T,<,I )$ is called \textit{frame for $\mathfrak{L}$}
if $( 
T,< )$ is a tree, and
$I$ is an indistinguishability function.
A $4$-tuple $\mathfrak{M}=
(  T,<,I,V )$ is called \textit{model for
$\mathfrak{L}$} if $(  T,<,I )$ is a frame for $\mathfrak{L}$ and
$V:PV\rightarrow\mathfrak{P}(\bigcup_{t\in T}(\{t\}\times\Pi_{\mathfrak{T},t}))$
is a function, called
\textit{evaluation}.
\end{definition}

\begin{notation} The following conventions are assumed:
\begin{enumerate}
\item Given a frame $\mathfrak{F}=(T,<,I)$,
$(t,\pi)\in\mathfrak{F}$ means that $(t,\pi)\in\bigcup_{t\in T}(\{t\}\times\Pi_t)$.
\item Given a model $\mathfrak{M}=(T,<,I,V)$,
$(t,\pi)\in\mathfrak{M}$ means that $(t,\pi)\in\bigcup_{t\in T}(\{t\}\times\Pi_t)$.
\end{enumerate}
\end{notation}

\begin{definition}\label{sem}
Given a model $\mathfrak{M}= (  T,<,I,V )$, a couple $(t,\pi)\in\mathfrak{M}$,
an atom $p$ and two formulas $\varphi$ and $\psi$, define
\begin{enumerate}
\item $\mathfrak{M},(t,\pi)\models p$ provided $(t,\pi)\in V(p)$.

\item $\mathfrak{M},(t,\pi)\models \neg\varphi$ provided
$\mathfrak{M},(t,\pi)\not\models \varphi$.

\item $\mathfrak{M},(t,\pi)\models \varphi\wedge\psi$ provided 
$\mathfrak{M},(t,\pi)\models \varphi$ and $\mathfrak{M},(t,\pi)\models
\psi$.

\item $\mathfrak{M},(t,\pi)\models G\varphi$ provided, for each
$h\in \pi$ and each $s\in h$ with $t<s$,
$\mathfrak{M},(s,[h]_{I_s})\models\varphi$.

\item $\mathfrak{M},(t,\pi)\models H\varphi$ provided, for each $h\in
\pi$ and each $s\in h$ with $s<t$, $\mathfrak{M},(s,[h]_{I_s})\models\varphi$.

\item $\mathfrak{M},(t,\pi)\models L\varphi$ provided, 
for each $\rho\in\Pi_t$, $\mathfrak{M},(t,\rho)\models\varphi$.
\end{enumerate}

Consider a formula $\varphi$.
Given a frame $\mathfrak{F}$, \textit{$\varphi$ is satisfiable in $\mathfrak{F}$}
provided there is an evaluation function $V$ and a couple $(t,\pi)\in\mathfrak{F}$ such that
$(\mathfrak{F},V),(t,\pi)\models\varphi$;
\textit{$\varphi$ is valid in $\mathfrak{F}$} provided,
for every evaluation function $V$ and couple $(t,\pi)\in\mathfrak{F}$,
$(\mathfrak{F},V),(t,\pi)\models\varphi$.
Given a model $\mathfrak{M}$, \textit{$\varphi$ is satisfiable in $\mathfrak{M}$}
provided there is a couple $(t,\pi)\in\mathfrak{M}$ such that
$\mathfrak{M},(t,\pi)\models\varphi$;
\textit{$\varphi$ is valid in $\mathfrak{M}$} provided,
for every couple $(t,\pi)\in\mathfrak{M}$,
$\mathfrak{M},(t,\pi)\models\varphi$.
\end{definition}

\subsection{A different view of the semantics}\label{alternative semantics}
In this section, a different view of the semantics for branching-time logics
with indistinguishability relations is presented.
\begin{definition} Let $\mathfrak{M}=(T,<,I,V)$ be a model and $(t,\pi)$, $(s,\rho)\in\mathfrak{M}$.
Define:
\begin{enumerate}
 \item $(t,\pi)\prec(s,\rho)$ provided $t<s$ and $\pi\supseteq\rho$
 \item $(t,\pi)\sim(s,\rho)$ provided  $t=s$ and $\pi,\rho\in I_t$ (iff $t=s$ and $\pi,\rho\in I_s$).
 \item $(t,\pi)\preceq(s,\rho)$ provided $(t,\pi)\prec(s,\rho)$ or $(t,\pi)=(s,\rho)$.
 \item $(t,\pi)\succ(s,\rho)$
 (resp. $(t,\pi)\succeq(s,\rho)$)
 provided
 $(s,\rho)\prec(t,\pi)$
 (resp. $(s,\rho)\preceq(t,\pi)$).
\end{enumerate}
\end{definition}

\begin{proposition}
Let $\mathfrak{M}=(T,<,I,V)$ be a model, $(t,\pi)\in\mathfrak{M}$ and $\varphi$ a formula
of $\mathfrak{L}$. Then:
\begin{enumerate}
\item\label{G} $\mathfrak{M},(t,\pi)\models G\varphi$ iff, for all $(s,\rho)\in\mathfrak{M}$ with
$(t,\pi)\prec(s,\rho)$, $\mathfrak{M},(s,\rho)\models \varphi$.
\item\label{H} $\mathfrak{M},(t,\pi)\models H\varphi$ iff, for all $(s,\rho)\in\mathfrak{M}$ with
$(s,\rho)\prec(t,\pi)$, $\mathfrak{M},(s,\rho)\models \varphi$.
\item\label{L} $\mathfrak{M},(t,\pi)\models L\varphi$ iff, for all $(s,\rho)\in\mathfrak{M}$ with
$(t,\pi)\sim(s,\rho)$, $\mathfrak{M},(s,\rho)\models \varphi$.
\end{enumerate}
\end{proposition}
\begin{proof}
 \ref{G}. Assume that $\mathfrak{M},(t,\pi)\models G\varphi$. Take any $(s,\rho)$ with $(t,\pi)\prec(s,\rho)$
 and any $h\in\rho$. Since $\pi\supseteq\rho$, $h\in\pi$.
 Then, since $t<s$, $\mathfrak{M},(s,[h]_{I_s})\models\varphi$. Thus, since $[h]_{I_s}=\rho$,
 $\mathfrak{M},(s,\rho)\models\varphi$. Hence, by arbitrariness of $(s,\rho)$,
 for all $(s,\rho)\in\mathfrak{M}$ with $(t,\pi)\prec(s,\rho)$, $\mathfrak{M},(s,\rho)\models \varphi$.
 
 Assume that, for all $(s,\rho)\in\mathfrak{M}$,
 $(t,\pi)\prec(s,\rho)$ entails $\mathfrak{M},(s,\rho)\models\varphi$.
 Take any $h\in\pi$ and any $s\in h$ with $t<s$. By ind.\:condition,
 $\pi\supseteq[h]_{I_s}$. Thus, $(t,\pi)\prec(s,[h]_{I_s})$.
 Hence, $\mathfrak{M},(s,[h]_{I_s})\models\varphi$. Then, by arbitrariness of $h$ and $s$,
 $\mathfrak{M},(t,\pi)\models G\varphi$.
 
 \ref{H}. Assume that $\mathfrak{M},(t,\pi)\models H\varphi$.
 Take any $(s,\rho)$ with $(s,\rho)\prec(t,\pi)$.
 Take any $h\in\pi$. 
 Then, since $s<t$, $\mathfrak{M},(s,[h]_{I_s})\models\varphi$.
 Moreover, since $\rho\supseteq\pi$, $h\in\rho$.
 Thus, $[h]_{I_s}=\rho$.
 Hence, $\mathfrak{M},(s,\rho)\models\varphi$.
 Therefore, by arbitrariness of $(s,\rho)$,
 for all $(s,\rho)\in\mathfrak{M}$ with $(s,\rho)\prec(t,\pi)$, $\mathfrak{M},(s,\rho)\models \varphi$.
 
 Assume that, for all $(s,\rho)\in\mathfrak{M}$,
 $(s,\rho)\prec(t,\pi)$ entails $\mathfrak{M},(s,\rho)\models\varphi$.
 Take any $h\in\pi$ and any $s\in h$ with $s<t$. By ind.\:condition,
 $[h]_{I_s}\supseteq\pi$. Thus, $(s,[h]_{I_s})\prec(t,\pi)$.
 Hence, $\mathfrak{M},(s,[h]_{I_s})\models\varphi$. Then, by arbitrariness of $h$ and $s$,
 $\mathfrak{M},(t,\pi)\models H\varphi$.
 
 \ref{L}. Assume that $\mathfrak{M},(t,\pi)\models L\varphi$. Take any $(s,\rho)$ with $(t,\pi)\sim(s,\rho)$.
 Then, since $t=s$ and $\pi,\rho\in I_t$, $\mathfrak{M},(s,\rho)\models \varphi$.
 Thus, by arbitrariness of $(s,\rho)$,
 for all $(s,\rho)\in\mathfrak{M}$ with $(t,\pi)\sim(s,\rho)$, $\mathfrak{M},(s,\rho)\models\varphi$.
 
 Assume that, for all $(s,\rho)\in\mathfrak{M}$,
 $(t,\pi)\sim(s,\rho)$ entails $\mathfrak{M},(s,\rho)\models\varphi$.
 Take any $\rho\in I_t$. Then, $(t,\pi)\sim(t,\rho)$.
 Thus, $\mathfrak{M},(t,\rho)\models\varphi$. Hence, by arbitrariness of $\rho$,
 $\mathfrak{M},(t,\pi)\models L\varphi$.
\end{proof}

\section{Bisimulation and p-morphism}\label{bsim and pmorph}
Here, a notion of p-morphism and a notion of bisimulation are given.
A number of preservation results are proven.

\begin{definition}\label{pmorphism}
 Let $\mathfrak{F}=(T,<,I)$ and $\mathfrak{F}'=(T',<',I')$ be two frames,
 $\prec$ (resp. $\prec'$) induced by $<$ (resp. $<'$) and $\sim$ (resp. $\sim'$)
 induced by $I$ (resp. $I'$).
 A function 
 $f:\bigcup_{t\in T}(\{t\}\times\Pi_t)
 \rightarrow\bigcup_{t'\in T'}(\{t'\}\times\Pi_{t'})$
 is called \textit{frame p-morphism
 from $\mathfrak{F}$ to $\mathfrak{F}'$} provided the following conditions hold:
 \begin{enumerate}
  \item[G-f.] For every $(t,\pi)$, $(s,\rho)\in\mathfrak{F}$, if $(t,\pi)\prec(s,\rho)$ then
  $f((t,\pi))\prec'f((s,\rho))$.
  \item[G-b.] For every $(t,\pi)\in\mathfrak{F}$, $(s,\rho)'\in\mathfrak{F}'$, if
  $f((t,\pi))\prec'(s,\rho)'$ then there is $(s,\rho)\in\mathfrak{F}$ such that
  $(t,\pi)\prec(s,\rho)$ and $f((s,\rho))=(s,\rho)'$.
  
  \item[H-b.] For every $(t,\pi)\in\mathfrak{F}$, $(s,\rho)'\in\mathfrak{F}'$, if
  $f((t,\pi))\succ'(s,\rho)'$ then there is $(s,\rho)\in\mathfrak{F}$ such that
  $(t,\pi)\succ(s,\rho)$ and $f((s,\rho))=(s,\rho)'$.
  
  \item[L-f.] For every $(t,\pi)$, $(s,\rho)\in\mathfrak{F}$, if $(t,\pi)\sim(s,\rho)$ then
  $f((t,\pi))\sim' f((s,\rho))$.
  \item[L-b.] For every $(t,\pi)\in\mathfrak{F}$, $(s,\rho)'\in\mathfrak{F}'$, if
  $f((t,\pi))\sim'(s,\rho)'$ then there is $(s,\rho)\in\mathfrak{F}$ such that
  $(t,\pi)\sim(s,\rho)$ and $f((s,\rho))=(s,\rho)'$.
 \end{enumerate}
 We say that two frames $\mathfrak{F}$ and $\mathfrak{F}'$ are \textit{p-morphic} provided there
 is a frame p-morphism from $\mathfrak{F}$ to $\mathfrak{F}'$.
\end{definition}

\begin{definition}
 Let $\mathfrak{M}=(T,<,I,V)$ and $\mathfrak{M}'=(T',<',I',V')$ be two models.
 A frame p-morphism from $(T,<,I)$ to $(T',<',I')$ is called \textit{model p-morphism
 from $\mathfrak{M}$ to $\mathfrak{M}'$} provided the following condition holds:
 \begin{enumerate}
  \item[PV.] For every $(t,\pi)\in\mathfrak{M}$,
  for every $p\in PV$, $(t,\pi)\in V(p)$ iff $f((t,\pi))\in V'(p)$.
  \end{enumerate}
 We say that two models $\mathfrak{M}$ and $\mathfrak{M}'$ are \textit{p-morphic} provided there
 is a model p-morphism from $\mathfrak{M}$ to $\mathfrak{M}'$.
\end{definition}

\begin{proposition}
Given two frame $\mathfrak{F}=(\mathfrak{F},<,I)$ and $\mathfrak{F}'=(\mathfrak{F}',<',I')$,
a function $f:\bigcup_{t\in T}(\{t\}\times\Pi_t)
 \rightarrow\bigcup_{t'\in T'}(\{t'\}\times\Pi_{t'})$ is a frame p-morphism
iff, for all $(t,\pi)\in \mathfrak{F}$ and $S\in\{\prec,\succ,\sim\}$,
$\{f((s,\rho))\in \mathfrak{F}'\,|\,(s,\rho)\in \mathfrak{F}, (t,\pi)S (s,\rho)\}=
\{(r,\sigma)\in \mathfrak{F}'\,|\,f((t,\pi))S'(r,\sigma)\}$.
\end{proposition}
\begin{proof}
Assume $f$ is a p-morphism. Take any $(t,\pi)\in \mathfrak{F}$ and $S\in\{\prec,\succ,\sim\}$.
Take any $f((s,\rho))\in\{f((s,\rho))\in \mathfrak{F}'\,|\,(s,\rho)\in \mathfrak{F}, (t,\pi)S (s,\rho)\}$.
Then, by condition G-f if $S$ is $\prec$ or $\succ$, or by condition L-f if $S$ is $\sim$,
$f((t,\pi))Sf((s,\rho))$. Thus $f((s,\rho))\in\{(r,\sigma)\in \mathfrak{F}'\,|\,f((t,\pi))S'(r,\sigma)\}$.
Take any $(r,\sigma)'\in\{(r,\sigma)\in \mathfrak{F}'\,|\,f((t,\pi))S'(r,\sigma)\}$.
Then, by condition G-b if $S$ is $\prec$, by condition H-b if $S$ is $\succ$,
or by condition L-b if $S$ is $\sim$, there is $(r,\sigma)\in \mathfrak{F}$ such that
$(t,\pi)S(r,\sigma)$ and $f((r,\sigma))=(r,\sigma)'$. Thus,
$(r,\sigma)'\in\{f((s,\rho))\in \mathfrak{F}'\,|\,(s,\rho)\in \mathfrak{F}, (t,\pi)S (s,\rho)\}$. Hence,
$\{f((s,\rho))\in \mathfrak{F}'\,|\,(s,\rho)\in \mathfrak{F}, (t,\pi)S (s,\rho)\}=
\{(r,\sigma)\in \mathfrak{F}'\,|\,f((t,\pi))S'(r,\sigma)\}$.

Assume, for every $(t,\pi)\in\mathfrak{F}$, every $S\in\{\prec,\succ,\sim\}$,
$\{f((s,\rho))\in \mathfrak{F}'\,|\,(s,\rho)\in \mathfrak{F}, (t,\pi)S (s,\rho)\}=
\{(r,\sigma)\in \mathfrak{F}'\,|\,f((t,\pi))S'(r,\sigma)\}$.
Take any $(t,\pi)$, $(s,\rho)\in\mathfrak{F}$ and any $S\in\{\prec,\sim\}$. Suppose
$(t,\pi)S(s,\rho)$.
Then,
$f((s,\rho))\in\{f((s,\rho))\in \mathfrak{F}'\,|\,(s,\rho)\in \mathfrak{F}, (t,\pi)S(s,\rho)\}$
Thus, $f((s,\rho))\in\{(r,\sigma)\in \mathfrak{F}'\,|\,f((t,\pi))S'(r,\sigma)\}$
Hence, $f((t,\pi))S'f((s,\rho))$ and
conditions G-f and L-f hold.
Moreover, take any $(t,\pi)\in\mathfrak{F}$, any $(s,\rho)'\in\mathfrak{F}'$
and $S\in\{\prec,\succ,\sim\}$. Suppose $f((t,\pi))S'(s,\rho)'$.
Then, $(s,\rho)'\in\{(r,\sigma)\in \mathfrak{F}'\,|\,f((t,\pi))S'(r,\sigma)\}$.
Thus, $(s,\rho)'\in\{f((s,\rho))\in \mathfrak{F}'\,|\,(s,\rho)\in \mathfrak{F}, (t,\pi)S (s,\rho)\}$.
Hence, there is $(s,\rho)\in\mathfrak{F}$, such that $f((s,\rho))=(s,\rho)'$ and $(t,\pi)S(s,\rho)$
and conditions G-b, H-b and L-b are satisfied. Therefore, $f$ is a frame p-morphism between
$\mathfrak{F}$ and $\mathfrak{F}'$.
\end{proof}

\begin{proposition}\label{p-morphism 1}
 Let $\mathfrak{M}=(T,<,I,V)$ and $\mathfrak{M}'=(T',<',I',V')$ be two models and
 $f$ a model p-morphism from $\mathfrak{M}$ to $\mathfrak{M'}$. Then, for any couple $(t,\pi)$
 of $\mathfrak{M}$ and any formula $\varphi$
 $$
 \mathfrak{M},(t,\pi)\models\varphi\mbox{ iff }\mathfrak{M}',f((t,\pi))\models\varphi.
 $$
\end{proposition}
\begin{proof}
 By induction on the complexity of $\varphi$. Suppose $\varphi$ is $p$, for arbitrary $p\in PV$.
 By condition PV, $\mathfrak{M},(t,\pi)\models p$ iff $\mathfrak{M}',f((t,\pi))\models p$.
 The boolean cases are easy.
 
 Suppose $\varphi$ is $G\psi$ (resp. $H\psi$, $L\psi$). Assume $\mathfrak{M},(t,\pi)\models G\psi$
 (resp. $\mathfrak{M},(t,\pi)\models H\psi$, $\mathfrak{M},(t,\pi)\models L\psi$).
 Consider any $(s,\rho)'\in\mathfrak{M}'$ with $f((t,\pi))\prec' (s,\rho)'$
 (resp. $f((t,\pi))\succ' (s,\rho)', f((t,\pi))\sim' (s,\rho)'$).
 By condition G-b (resp. H-b, L-b), there is $(s,\rho)\in\mathfrak{M}$ with $(t,\pi)\prec(s,\rho)$
 (resp. $(t,\pi)\succ(s,\rho)$, $(t,\pi)\sim(s,\rho)$) and
 $f((s,\rho))=(s,\rho)'$. Since, $\mathfrak{M},(t,\pi)\models G\psi$
 (resp. $\mathfrak{M},(t,\pi)\models H\psi$, $\mathfrak{M},(t,\pi)\models L\psi$),
 $\mathfrak{M},(s,\rho)\models\psi$. Then, by i.h., $\mathfrak{M}',(s,\rho)'\models\psi$.
 Thus, by arbitrariness of $(s,\rho)'$, $\mathfrak{M}',f((t,\pi))\models G\psi$
 (resp. $\mathfrak{M}',f(t,\pi)\models H\psi$, $\mathfrak{M}',f(t,\pi)\models L\psi$).
 
 Assume $\mathfrak{M}',f((t,\pi))\models G\psi$
 (resp. $\mathfrak{M}',f((t,\pi))\models H\psi$, $\mathfrak{M}',f((t,\pi))\models L\psi$).
 Take any $(s,\rho)\in\mathfrak{M}$ with
 $(t,\pi)\prec(s,\rho)$ (resp. $(t,\pi)\succ(s,\rho)$, $(t,\pi)\sim(s,\rho)$).
 By condition G-f (resp. G-f, L-f), $f((t,\pi))\prec'f((s,\rho))$
 (resp. $f((t,\pi))\succ'f((s,\rho))$, $f((t,\pi))\sim'f((s,\rho))$). As
 $\mathfrak{M}',f((t,\pi))\models G\psi$
 (resp. $\mathfrak{M}',f((t,\pi))\models H\psi$, $\mathfrak{M}',f((t,\pi))\models L\psi$),
 $\mathfrak{M}',f((s,\rho))\models\psi$.
 Then, by i.h, $\mathfrak{M},(s,\rho)\models\psi$. Thus, by arbitrariness of $(s,\rho)$,
 $\mathfrak{M},(t,\pi)\models G\psi$
 (resp. $\mathfrak{M},(t,\pi)\models H\psi$, $\mathfrak{M},(t,\pi)\models L\psi$).
\end{proof}

\begin{proposition}
Let $\mathfrak{F}=(T,<,I)$ and $\mathfrak{F}'=(T',<',I')$ be two frames.
Suppose that $\mathfrak{F}'$ is a p-morphic image of $\mathfrak{F}$.
Then, for every formula $\varphi$, if $\varphi$ is valid in $\mathfrak{F}$
then $\varphi$ is valid in $\mathfrak{F}'$. 
\end{proposition}
\begin{proof}
Let $\varphi$ be a formula valid in $\mathfrak{F}$.
Let $f$ be a surjective frame p-morphism from $\mathfrak{F}$ to
$\mathfrak{F}'$.
Let $V':PV\rightarrow\mathfrak{P}(\bigcup_{t'\in T'}(\{t'\}\times\Pi_{\mathfrak{T'},t'}))$
be an arbitrary evaluation function.
Take an arbitrary $(t,\pi)'\in\mathfrak{T}'$.
Define
$V:PV\rightarrow\mathfrak{P}(\bigcup_{t\in T}(\{t\}\times\Pi_{\mathfrak{T},t}))$
by, for every $p\in PV$,
$V(p)=\{(t,\pi)\in\mathfrak{F}\,|\,f((t,\pi))\in V'(p)\}$.
Then, $f$ satisfies condition PV.
Thus, $f$ is a model p-morphism from $(\mathfrak{F},V)$
to $(\mathfrak{F}',V')$.
Since $f$ is surjective,
there is $(t,\pi)\in\mathfrak{F}$ such that $f((t,\pi))=(t,\pi)'$.
Since $\varphi$ is valid in $\mathfrak{F}$, $(\mathfrak{F},V),(t,\pi)\models\varphi$.
Thus, by prop. \ref{p-morphism 1}, $(\mathfrak{F}',V'),(t,\pi)'\models\varphi$.
Therefore, by arbitrariness of $V'$ and $(t,\pi)'$, $\varphi$ is valid in $\mathfrak{F}'$.
\end{proof}

\begin{definition}\label{bisimulation}
Let $\mathfrak{M}=(T,<,I,V)$ and $\mathfrak{M}'=(T',<',I',V')$ be models.
Let $(t,\pi)\in\mathfrak{M}$ and $(t,\pi)'\in\mathfrak{M}'$.
A \textit{bisimulation between $(M,(t,\pi))$ and $(M',(t,\pi)')$} is a
relation $B\subseteq\bigcup_{t\in T}(\{t\}\times\Pi_{\mathfrak{T},t})\times
\bigcup_{t'\in T'}(\{t'\}\times\Pi_{\mathfrak{T}',t'})$ satisfying:
\begin{enumerate}
\item[B.] $(t,\pi)B(t,\pi)'$.
\end{enumerate}
and, for every $(s,\rho)\in\mathfrak{M}$ and $(s,\rho)'\in\mathfrak{M}'$ such that $(s,\rho)B(s,\rho)'$:
\begin{enumerate}
\item[PV.] for every $p\in PV$, $(s,\rho)\in V(p)$ iff $(s,\rho)'\in V'(p)$.
\item[G-f.] For every $(r,\sigma)\in\mathfrak{M}$, if $(s,\rho)\prec(r,\sigma)$ then
there is $(r,\sigma)'\in\mathfrak{M}'$ with $(s,\rho)'\prec'(r,\sigma)'$ and $(r,\sigma)B(r,\sigma)'$.
\item[G-b.] For every $(r,\sigma)'\in\mathfrak{M}'$, if $(s,\rho)'\prec'(r,\sigma)'$,
there is $(r,\sigma)\in\mathfrak{M}$ such that
$(s,\rho)\prec(r,\sigma)$ and $(r,\sigma)B(r,\sigma)'$.
\item[H-f.] For every $(r,\sigma)\in\mathfrak{M}$, if $(s,\rho)\succ(r,\sigma)$ then
there is $(r,\sigma)'\in\mathfrak{M}'$ with $(s,\rho)'\succ'(r,\sigma)'$ and $(r,\sigma)B(r,\sigma)'$.
\item[H-b.] For every $(r,\sigma)'\in\mathfrak{M}'$, if $(s,\rho)'\succ'(r,\sigma)'$,
there is $(r,\sigma)\in\mathfrak{M}$ such that
$(s,\rho)\succ(r,\sigma)$ and $(r,\sigma)B(r,\sigma)'$.
\item[L-f.] For every $(r,\sigma)\in\mathfrak{M}$, if $(s,\rho)\sim(r,\sigma)$ then
there is $(r,\sigma)'\in\mathfrak{M}'$ with $(s,\rho)'\sim'(r,\sigma)'$ and $(r,\sigma)B(r,\sigma)'$.
\item[L-b.] For every $(r,\sigma)'\in\mathfrak{M}'$, if $(s,\rho)'\sim'(r,\sigma)'$,
there is $(r,\sigma)\in\mathfrak{M}$ such that
$(s,\rho)\sim(r,\sigma)$ and $(r,\sigma)B(r,\sigma)'$.
\end{enumerate}
Given two models $\mathfrak{M}$ and $\mathfrak{M}'$, and $(t,\pi)\in\mathfrak{M}$ and
$(t,\pi)'\in\mathfrak{M}'$, we say that $(\mathfrak{M},(t,\pi))$ and
$(\mathfrak{M}',(t,\pi)')$ are \textit{bisimilar} provided there is a bisimulation between
$(\mathfrak{M},(t,\pi))$ and
$(\mathfrak{M}',(t,\pi)')$.
\end{definition}

\begin{proposition}\label{bisimulation 1}
Let $\mathfrak{M}$ and $\mathfrak{M}'$ be two models, and $(t,\pi)\in\mathfrak{M}$ and
$(t,\pi)'\in\mathfrak{M}'$ such that there is a bisimulation $B$ between $(\mathfrak{M},(t,\pi))$ and
$(\mathfrak{M}',(t,\pi)')$. Then, for every formula $\varphi$,
$\mathfrak{M},(t,\pi)\models\varphi$ iff $\mathfrak{M}',(t,\pi)'\models\varphi$.
\end{proposition}
\begin{proof}
 By induction on the complexity of $\varphi$. Suppose $\varphi$ is $p$, for arbitrary $p\in PV$.
 By condition PV, $\mathfrak{M},(t,\pi)\models p$ iff $\mathfrak{M}',(t,\pi)'\models p$.
 The boolean cases are easy.
 
 Suppose $\varphi$ is $G\psi$ (resp. $H\psi$, $L\psi$). Assume $\mathfrak{M},(t,\pi)\models G\psi$
 (resp. $\mathfrak{M},(t,\pi)\models H\psi$, $\mathfrak{M},(t,\pi)\models L\psi$).
 Consider any $(s,\rho)'\in\mathfrak{M}'$ with $(t,\pi)'\prec'(s,\rho)'$
 (resp. $(t,\pi)'\succ'(s,\rho)', (t,\pi)'\sim'(s,\rho)'$).
 By condition G-b (resp. H-b, L-b), there is $(s,\rho)\in\mathfrak{M}$ with $(t,\pi)\prec(s,\rho)$
 (resp. $(t,\pi)\succ(s,\rho)$, $(t,\pi)\sim(s,\rho)$) and
 $(s,\rho)B(s,\rho)'$. Since, $\mathfrak{M},(t,\pi)\models G\psi$
 (resp. $\mathfrak{M},(t,\pi)\models H\psi$, $\mathfrak{M},(t,\pi)\models L\psi$),
 $\mathfrak{M},(s,\rho)\models\psi$. Then, by i.h., $\mathfrak{M}',(s,\rho)'\models\psi$.
 Thus, by arbitrariness of $(s,\rho)'$, $\mathfrak{M}',(t,\pi)'\models G\psi$
 (resp. $\mathfrak{M}',(t,\pi)'\models H\psi$, $\mathfrak{M}',(t,\pi)'\models L\psi$).
 
 Assume $\mathfrak{M}',(t,\pi)'\models G\psi$
 (resp. $\mathfrak{M}',(t,\pi)'\models H\psi$, $\mathfrak{M}',(t,\pi)'\models L\psi$).
 Take any $(s,\rho)\in\mathfrak{M}$ with
 $(t,\pi)\prec(s,\rho)$ (resp. $(t,\pi)\succ(s,\rho)$, $(t,\pi)\sim(s,\rho)$).
 By condition G-f (resp. H-f, L-f), there is $(s,\rho)'\in\mathfrak{M}'$ such that
 $(t,\pi)'\prec'(s,\rho)'$
 (resp. $(t,\pi)'\succ'(s,\rho)'$, $(t,\pi)'\sim'(s,\rho)'$) and $(s,\rho)B(s,\rho)'$.
 As  $\mathfrak{M}',(t,\pi)'\models G\psi$
 (resp. $\mathfrak{M}',(t,\pi)'\models H\psi$, $\mathfrak{M}',(t,\pi)'\models L\psi$),
 $\mathfrak{M}',(s,\rho)'\models\psi$.
 Then, by i.h, $\mathfrak{M},(s,\rho)\models\psi$. Thus, by arbitrariness of $(s,\rho)$,
 $\mathfrak{M},(t,\pi)\models G\psi$
 (resp. $\mathfrak{M},(t,\pi)\models H\psi$, $\mathfrak{M},(t,\pi)\models L\psi$).
\end{proof}

\section{Adding the weak future operator $F$}\label{adding F}
In this section, the language is enriched with the ``weak future operator'' $F$.
A notion of p-morphism and a notion of bisimulation are given.
A number of preservation results are proven.
\begin{definition}
$\mathfrak{L}_F=\mathfrak{L}\cup\{F\}$ is called \textit{language}.
\textit{Formulas} are strings of elements of the language
built up recursively according to the rules of def. \ref{language} plus:
\begin{enumerate}
\item[3.] If $\varphi$ is a formula,
$(F\varphi)$ is a formula.
\end{enumerate}
The abbreviations of def. \ref{language} are assumed. In addition, $g$ abbreviates $\neg F\neg$.
The usual precedence rules among operators are assumed.
\end{definition}

\begin{definition}
 Semantics is defined as in def. \ref{sem}, plus
\begin{enumerate}
\item[7.] $\mathfrak{M},(t,\pi)\models F\varphi$ provided, for each
$h\in \pi$, there is $s\in h$ such that $t<s$ and
$\mathfrak{M},(s,[h]_{I_s})\models\varphi$.
\end{enumerate}
\end{definition}

\begin{definition}
Given a function $f:\bigcup_{t\in T}(\{t\}\times\Pi_t)
 \rightarrow\bigcup_{t'\in T'}(\{t'\}\times\Pi_{t'})$,
 for $i\in\{1,2\}$, $f_i$ denotes the $i$-th component of $f$.
\end{definition}

\begin{definition}
 Let $\mathfrak{F}=(T,<,I)$ and $\mathfrak{F}'=(T',<',I')$ be two frames.
 A function 
 $f:\bigcup_{t\in T}(\{t\}\times\Pi_t)
 \rightarrow\bigcup_{t'\in T'}(\{t'\}\times\Pi_{t'})$
 is called \textit{frame p-morphism
 from $\mathfrak{F}$ to $\mathfrak{F}'$} provided, in addition to G-f, G-b, H-b, L-f, L-b of def. \ref{pmorphism},
 the following conditions hold:
 \begin{enumerate}
  \item[F-f.] For every $(t,\pi)\in\mathfrak{M}$, for every $h'\in f_2((t,\pi))$, there is $h\in\pi$
  such that, for every $s\in h$ with $t<s$, there is $s'\in h'$ with 
  $f_1((t,\pi))<'s'$ and $(s',[h']_{I'_{s'}})=f((s,[h]_{I_s}))$.
  \item[F-b.] For every $(t,\pi)\in\mathfrak{M}$, for every $h\in \pi$, there is $h'\in f_2((t,\pi))$
  such that, for every $s'\in h'$ with $f_1((t,\pi))<'s'$, there is $s\in h$ such that
  $t<s$ and $f((s,[h]_{I_s}))=(s',[h']_{I'_{s'}})$. 
 \end{enumerate}
 We say that two frames $\mathfrak{F}$ and $\mathfrak{F}'$ are \textit{p-morphic} provided there
 is a frame p-morphism from $\mathfrak{F}$ to $\mathfrak{F}'$.
\end{definition}
\begin{definition}
 Let $\mathfrak{M}=(T,<,I,V)$ and $\mathfrak{M}'=(T',<',I',V')$ be two models.
 A frame p-morphism from $(T,<,I)$ to $(T',<',I')$ is called \textit{model p-morphism
 from $\mathfrak{M}$ to $\mathfrak{M}'$} provided $PV$ holds.
\end{definition}

\begin{proposition}\label{p-morphism F}
 Let $\mathfrak{M}=(T,<,I,V)$ and $\mathfrak{M}'=(T',<',I',V')$ be two models and
 $f$ a model p-morphism from $\mathfrak{M}$ to $\mathfrak{M'}$. Then, for any couple $(t,\pi)$
 of $\mathfrak{M}$ and any formula $\varphi$
 $$
 \mathfrak{M},(t,\pi)\models\varphi\mbox{ iff }\mathfrak{M}',f((t,\pi))\models\varphi.
 $$
\end{proposition}
\begin{proof}
 By induction on the complexity of $\varphi$.
 A part from the case in which $\varphi$ is $F\psi$, for some formula $\psi$, the proof goes
 as the proof of prop. \ref{p-morphism 1}.
 
 Suppose $\varphi$ is $F\psi$.
 Suppose $\mathfrak{M}',f((t,\pi))\not\models F\psi$. 
 Then, there is $h'\in f_2((t,\pi))$ such that, for every $s'\in h'$ with $f_1((t,\pi))<'s'$, $\mathfrak{M}',(s',[h']_{I_{s'}})\not\models\psi$.
 Thus, by condition F-f, there is $h\in\pi$ such that, for every $s\in h$ with $t<s$, there is
 $s'\in h'$ with $f_1((t,\pi))<s'$ and $(s',[h']_{I'_{s'}})=f((s,[h]_{I_s}))$.
 Therefore, by i.h., for every $s\in h$ with $t<s$, $\mathfrak{M},(s,[h]_{I_s})\not\models\psi$.
 Hence, by def., $\mathfrak{M},(t,\pi)\not\models F\psi$.
 
 Suppose $\mathfrak{M},(t,\pi)\not\models F\psi$.
 Then, there is $h\in\pi$ such that, for every $s\in h$ with $t<s$, $\mathfrak{M},(s,[h]_{I_{s}})\not\models\psi$.
 Thus, by condition F-b, there is $h'\in f_2((t,\pi))$ such that, for every $s'\in h'$ with $f_1((t,\pi))<s'$, there is
 $s\in h$ with $t<s$ and $f((s,[h]_{I_s}))=(s',[h']_{I'_{s'}})$.
 Therefore, by i.h., for every $s'\in h'$ with $f_1((t,\pi))<s'$, $\mathfrak{M}',(s',[h']_{I_{s'}})\not\models\psi$.
 Hence, by def., $\mathfrak{M},f((t,\pi))\not\models F\psi$.
\end{proof}

\begin{definition}
Let $\mathfrak{M}=(T,<,I,V)$ and $\mathfrak{M}'=(T',<',I',V')$ be models.
Let $(t,\pi)\in\mathfrak{M}$ and $(t,\pi)'\in\mathfrak{M}'$.
A \textit{bisimulation between $(M,(t,\pi))$ and $(M',(t,\pi)')$} is a
relation $B\subseteq\bigcup_{t\in T}(\{t\}\times\Pi_{\mathfrak{T},t})\times
\bigcup_{t'\in T'}(\{t'\}\times\Pi_{\mathfrak{T}',t'})$ satisfying, in addition to
B, PV, G-f, G-b, H-f, H-b, L-f, L-b of def. \ref{bisimulation}, for every $(s,\rho)\in\mathfrak{M}$ and $(s',\rho')\in\mathfrak{M}'$ such that $(s,\rho)B(s',\rho')$:
\begin{enumerate}
  \item[F-f.] For every $h'\in \rho'$, there is $h\in\rho$
  such that, for every $r\in h$ with $s<r$, there is $r'\in h'$ with 
  $s'<'r'$ and $(r,[h]_{I_r})B(r',[h']_{I'_{r'}})$.
  \item[F-b.] For every $h\in \rho$, there is $h'\in\rho'$
  such that, for every $r'\in h'$ with $s'<'r'$, there is $r\in h$ with 
  $s<r$ and $(r,[h]_{I_r})B(r',[h']_{I'_{r'}})$.
\end{enumerate}
Given two models $\mathfrak{M}$ and $\mathfrak{M}'$, and $(t,\pi)\in\mathfrak{M}$ and
$(t,\pi)'\in\mathfrak{M}'$, we say that $(\mathfrak{M},(t,\pi))$ and
$(\mathfrak{M}',(t,\pi)')$ are \textit{bisimilar} provided there is a bisimulation between
$(\mathfrak{M},(t,\pi))$ and
$(\mathfrak{M}',(t,\pi)')$.
\end{definition}

\begin{proposition}
Lets $\mathfrak{M}$ and $\mathfrak{M}'$ be two models, and $(t,\pi)\in\mathfrak{M}$ and
$(t',\pi')\in\mathfrak{M}'$ such that there is a bisimulation $B$ between $(\mathfrak{M},(t,\pi))$ and
$(\mathfrak{M}',(t',\pi'))$. Then, for every formula $\varphi$,
$\mathfrak{M},(t,\pi)\models\varphi$ iff $\mathfrak{M}',(t',\pi')\models\varphi$.
\end{proposition}
\begin{proof}
 By induction on the complexity of $\varphi$.
 A part from the case in which $\varphi$ is $F\psi$, for some formula $\psi$, the proof goes
 as the proof of prop. \ref{p-morphism 1}.
 
 Suppose $\varphi$ is $F\psi$.
 Suppose $\mathfrak{M}',(t',\pi')\not\models F\psi$. 
 Then, there is $h'\in \pi'$ such that, for every $s'\in h'$ with $t'<s'$, $\mathfrak{M}',(s',[h']_{I'_{s'}})\not\models\psi$.
 Thus, by condition F-f, there is $h\in\pi$ such that, for every $s\in h$ with $t<s$, there is
 $s'\in h'$ with $t'<'s'$ and $(s,[h]_{I_s})B(s',[h']_{I'_{s'}})$.
 Therefore, by i.h., for every $s\in h$ with $t<s$, $\mathfrak{M},(s,[h]_{I_s})\not\models\psi$.
 Hence, by def., $\mathfrak{M},(t,\pi)\not\models F\psi$.
 
 Suppose $\mathfrak{M},(t,\pi)\not\models F\psi$.
 Then, there is $h\in\pi$ such that, for every $s\in h$ with $t<s$, $\mathfrak{M},(s,[h]_{I_{s}})\not\models\psi$.
 Thus, by condition F-b, there is $h'\in \pi'$ such that, for every $s'\in h'$ with $t'<'s'$, there is
 $s\in h$ with $t<s$ and $(s,[h]_{I_s})B(s',[h']_{I'_{s'}})$.
 Therefore, by i.h., for every $s'\in h'$ with $t'<'s'$, $\mathfrak{M}',(s',[h']_{I'_{s'}})\not\models\psi$.
 Hence, by def., $\mathfrak{M}',(t',\pi')\not\models F\psi$.
\end{proof}

\bibliographystyle{plain}

\end{document}